\documentclass[11pt]{amsart}

\usepackage{amssymb, amsmath}
\usepackage{amsfonts}
\usepackage{float}
\usepackage{multirow}
\usepackage{graphics}
\usepackage{color}
\usepackage{graphicx}
\usepackage{enumitem}
\usepackage[usenames,dvipsnames]{xcolor}
\usepackage[section]{placeins}
\numberwithin{equation}{section}

\renewcommand\Re{\operatorname{Re}}
\renewcommand\Im{\operatorname{Im}}

\newtheorem{assumption}{Assumption}

\newtheorem{theorem}{Theorem}

\begin{document}
\title{Dynamic Transitions of Quasi-Geostrophic Channel Flow}

\author[Dijkstra]{Henk Dijkstra}
\address[HD]{Institute for Marine and Atmospheric research Utrecht
Department of Physics and Astronomy
Utrecht University
Princetonplein 5, 3584 CC Utrecht,
The Netherlands}
\email{H.A.Dijkstra@uu.nl}

\author[Sengul]{Taylan Sengul}
\address[TS]{Department of Mathematics, Yeditepe University, 34750 Istanbul, Turkey}
\email{taylansengul@gmail.com}

\author[Shen]{Jie Shen}
\address[JS]{Department of Mathematics, Purdue University, West Lafayette, IN 47907}
\email{shen7@purdue.edu}

\author[Wang]{Shouhong Wang}
\address[SW]{Department of Mathematics,
Indiana University, Bloomington, IN 47405}
\email{showang@indiana.edu, http://www.indiana.edu/~fluid}

\begin{abstract}
The main aim of this paper is to describe the dynamic transitions in flows described by the two-dimensional, barotropic vorticity equation in a periodic zonal channel. In \cite{CGSW03}, the
existence of a Hopf bifurcation in this model as the Reynolds number crosses a critical value was proven. In this paper, we extend the results in \cite{CGSW03} by addressing the stability problem of the bifurcated periodic solutions. Our main result is the explicit expression of a non-dimensional number $\gamma$ which controls the transition behavior. We prove that depending on $\gamma$,
the  modeled flow  exhibits either a continuous (Type I) or catastrophic (Type II) transition. Numerical evaluation of $\gamma$ for a physically realistic region of parameter space  suggest that a catastrophic transition is preferred in this flow.
\end{abstract}
\maketitle

\section{Introduction}
Climate variability exhibits recurrent large-scale patterns which are
directly  linked to dynamical processes represented in the governing
dissipative dynamical system \cite{DijkstraB2000,DG05, ghil00}. The
study of  the persistence of these patterns and the transitions between
them  also play a crucial role in understanding climate change and in
interpreting future climate projections \cite{IPCC2013}.

Current climate models used for developing such  projections are
based on the conservation laws of fluid mechanics and  consist of
systems of  nonlinear partial differential equations (PDEs).  These
can be put into the perspective of infinite-dimensional dissipative systems
exhibiting  large-dimensional attractors. The global attractor is a mathematical
object strongly connected to  the overall dissipation in the system. Climate
variability is, however,  often associated with  dynamic transitions between
different regimes, each  represented by local attractors.

There are many examples of climate phenomena where such transitions
have been investigated numerically, such as the transition to oscillatory
behavior in models of the El Ni\~no/Southern Oscillation  phenomenon in the
equatorial Pacific, the  transitions between different mean flow patterns of
the Kuroshio Current in the North Pacific and the transitions between
blocked and zonal flows in the midlatitude atmosphere (see, e.g.,
\cite{Dijkstra2013}). However, rigorous  mathematical results on the type
of the transitions in these systems of PDEs are extremely scarce.

This paper arises out of a research program to generate rigorous  mathematical
results  on climate variability developed from the viewpoint of  dynamical
transitions.   We have  shown \cite{ptd} that the transitions of all dissipative
systems can  be  classified into three  classes: continuous, catastrophic
and random, which  correspond to very different dynamical transition behavior
of the system.

We here focus on the dynamic transitions in flows described by one of the
cornerstone dynamical models in both oceanic and atmospheric dynamics,
the two-dimensional, dimensionless barotropic vorticity equation given by
\begin{equation} \label{pre-main}
\frac{\partial \Delta \psi}{\partial t} + \epsilon J(\psi,\Delta \psi) + \frac{\partial \psi}{\partial x}  =
E \Delta^2 \psi + \alpha_\tau \sin \pi y,
\end{equation}
where $\Delta$ is the Laplacian operator, $J(F, G) = (\partial F/\partial x) (\partial G / \partial y) - (\partial F / \partial y) (\partial G / \partial x) $ is the advection operator and $\psi$ the
geostrophic stream function. The equation (\ref{pre-main}) describes flows
with a typical length scale  $L$ on a mid-latitude beta-plane with a central latitude
$\theta_0$ and  a planetary vorticity gradient $\beta_0$.  It can be derived from the primitive
equations by  the so-called quasi-geostrophic (QG) approximation, which assumes a dominant   balance  between the Coriolis force and the pressure gradient force  \cite{pedlosky87}.

To derive (\ref{pre-main}), time, length and stream function were  non-dimensionalised
with $1/(\beta_0 L)$, $L$ and $U L$, where $U$ is a characteristic  horizontal velocity.
The positive parameters $\epsilon$ and $E$ are Rossby and Ekman numbers,
respectively, given by
\[
\epsilon = \frac{U}{\beta_0 L^2} ~ ; ~ E = \frac{A_H}{\beta_0 L^3}
\]
where $A_H$ is the lateral friction coefficient.  The Reynolds number $R$ is defined
as
\begin{equation} \label{RepsilonE}
R = \frac{\epsilon}{E} = \frac{U L}{A_H}
\end{equation}

The forcing term $\alpha_\tau  \sin \pi y$ in \eqref{pre-main} may represent the transfer of
angular momentum into midlatitudes due to tropical Hadley cell  in an atmospheric model.
In this case, the magnitude of the velocity can be scaled such that $\alpha_\tau = 1$.
In the ocean case, the forcing term represents the dimensionless wind stress
\[
  \tau = \frac{\alpha_\tau}{\pi}(\cos\pi y,0)
\]
Such a wind stress mimics the annually averaged zonal wind distribution
over the North Atlantic and North Pacific with westerly (i.e. eastward) winds over the midlatitudes
and easterlies in the tropics and polar latitudes. When the dimensionless wind stress has a
magnitude $\tau_0$ and the ocean basin a  depth $D$ and the water a density $\rho$, the
factor $\alpha_\tau$ is given by
\[
\alpha_\tau = \frac{\tau_0}{\rho D L \beta_0 U}
\]
In this case, we can choose $U$ such that $\alpha_\tau = 1$, which is usually referred
to as the Sverdrup scaling.  In both ocean and atmosphere cases the equation (\ref{pre-main})
has two free parameters  (out of the three $\epsilon$, $E$ and $R$) which we choose here
as $R$ and $E$.

We consider flows in a so-called zonal channel of length $2/a$ with walls bounding
the domain  at $y = \pm 1$ and periodic conditions in zonal direction. The equation
\eqref{pre-main} is therefore supplemented with boundary conditions
\begin{equation} \label{BC}
\begin{aligned}
& \psi \mid_{x = 0} = \psi \mid_{x = 2/a}, \\
%& \phi \;\text{ is periodic in } x \text{ with period } 2/a,\\
& \psi \mid_{y = \pm 1} =  \frac{\partial^2 \psi}{\partial y^2} \mid_{y = \pm 1} = 0.
\end{aligned}
\end{equation}

The equations \eqref{pre-main} with \eqref{BC} admit the following steady state
\begin{equation} \label{steady state}
\psi_0 = - \frac{1}{\pi^4 E} \sin\pi y,
\end{equation}
which represents a zonal jet with zonal velocity field $u_0  =   - \partial \psi_0 / \partial y =
1/(\pi^3 E) \cos \pi y$.

It is shown in \cite{CGSW03} that for any $a\geq \sqrt{3}/2$, $\psi_0$ is both
linearly and nonlinearly stable. Also, there is an $\alpha_0$ with $\sqrt{3}/4 < \alpha_0 < \sqrt{3}/2$ such that for any $\sqrt{3}/4 \leq a \leq  \alpha_0$, there is a critical Reynolds number $R_0>0$ depending on $a$ such that a simple pair of complex conjugate eigenvalues cross the imaginary axis as the Reynolds number $R$
crosses $R_0$, leading to the existence of Hopf bifurcation at the critical Reynolds number.
However,  the stability of the bifurcated periodic solutions and the dynamic transition behavior
near $R_0$  are  so far unknown. The main difficulty  is caused by the lack of explicit
analytical form of the eigenfunctions.

The main objective of this article is to investigate the dynamic transition and the stability of
 the basic state \eqref{steady state} as the Reynolds number crosses a critical threshold
 $R_0$. The main result we obtain is that the dynamic transition from this state to new states
 is either continuous (Type-I) or catastrophic (Type-II), and is determined  by  the sign of a
 computable  parameter $\gamma$ given by (\ref{gamma}). Our numerical investigations
 indicate that in a physically relevant parameter regime, only  catastrophic transitions occur.

\section{Main Theorem}

Throughout $\Re z$, $\Im z$, $\overline{z}$ will denote the real part, imaginary part and conjugate of a complex number $z$.  $D= \frac{d}{dy}$ is the derivative operator, $\Omega = (0, 2/a) \times (-1, 1) \subset \mathbb{R}^2$ and $(\cdot, \cdot)$ is the $L^2(\Omega)$ inner product.

\subsection{Functional setting}
Considering the deviation $\psi^{\prime} = \psi - \psi_0$ from the basic steady state
\eqref{steady state} and omitting the primes, we obtain from \eqref{pre-main},
\begin{equation} \label{main}
\frac{\partial \Delta \psi}{\partial t} + \epsilon J(\psi,\Delta \psi) = -\frac{R}{\pi^3} \cos \pi y (\Delta \frac{\partial \psi}{\partial x} +\pi^2 \frac{\partial \psi}{\partial x}) -\frac{\partial \psi}{\partial x} + E \Delta^2 \psi.
\end{equation}

Note that \eqref{pre-main} can also be formulated in velocity ${\bf u} = (- \frac{\partial \psi}{\partial y}, \frac{\partial \psi}{\partial x})$ and pressure $p$ as
\begin{equation} \label{primitive-formulation}
\begin{aligned}
& \frac{\partial {\bf u}}{\partial t} + \epsilon ({\bf u} \cdot \nabla) {\bf u} + f {\bf k} \times {\bf u} = E \Delta {\bf u} - \nabla p + {\bf \tau}, \\
& \nabla \cdot {\bf u} = 0.
\end{aligned}
\end{equation}
Here $f$ is the dimensionless Coriolis parameter on a midlatitude beta plane
that gives rise to the $\partial \psi / \partial x$ term in \eqref{pre-main}; ${\bf k}$
is the unit vector in the $z$-direction.

Using the formulation \eqref{primitive-formulation}, we can write the problem in the following abstract form
\begin{equation} \label{main functional equation}
\frac{d {\bf u}}{dt} = L {\bf u} + G({\bf u}),
\end{equation}
where $L:H_1 \rightarrow H$ is the linear operator, $G:H_1 \rightarrow H$ is the nonlinear operator, and
\begin{equation*}
\begin{aligned}
& H =  \{ {\bf u} = (u,v) \in (L_2(\Omega))^2 \mid & v \mid_{y = \pm 1} = 0, \, \text{div} {\bf u} = 0, {\bf u} \mid_{x = 0} = {\bf u} \mid_{x = 2/a} \} \\
& H_1 = H \cap (H^1(\Omega))^2.
\end{aligned}
\end{equation*}

The eigenvalue problem for the linearized equation of \eqref{main} reads
\begin{equation} \label{lin_eig_prob}
E \Delta^2 \psi + \frac{R}{\pi^3} \cos(\pi y) (\Delta \psi_x +\pi^2 \psi_x) - \psi_x = \beta \Delta \psi
\end{equation}
with boundary conditions \eqref{BC}.

Since the solution $\psi$ is periodic in $x$ with period $2/a$, we can expand $\psi$ in Fourier series, so for the $m$-mode of the expansion, we can write
\begin{equation} \label{ansatz}
\psi = e^{i \alpha_m x} Y(y), \quad \alpha_m = a m \pi.\end{equation}
where $Y(y)$ satisfies the boundary condition $Y(\pm 1)=D^2Y(\pm 1)=0$.
Plugging \eqref{ansatz} into \eqref{lin_eig_prob}, we obtain a sequence of one-dimensional problem:
\begin{equation}\label{1d linear equation for Y}
  E(D^2 - \alpha_m^2)^2 Y + i\alpha_m \left( \frac{R}{\pi^3} \cos(\pi y) (D^2 -\alpha_m^2 + \pi^2) - 1 \right)Y
  =\beta(D^2 -\alpha_m^2)Y.
\end{equation}

The eigenvectors of \eqref{lin_eig_prob} are $\psi_{m, j} = e^{i \alpha_m x} Y_{m, j}(y)$ where $Y_{m, j}$ are the eigenvectors of \eqref{1d linear equation for Y} corresponding to the eigenvalues $\beta_{m, j} \in \mathbb{C}$ where $m \in \mathbb{Z}$, $j = 1, 2, 3, \dots$. Moreover $\beta_{m, j}$ are ordered so that $\Re \beta_{m, j} \le \Re \beta_{m, k}$ if $j > k$. Also $\beta_{-m, j} = \overline{\beta_{m, j}}$ and we can take $\psi_{-m, j} = \overline{\psi_{m, j}}$. In particular $\psi_{0, j}$ and $\beta_{0, j}$ are real.

We also need to consider the eigenvalue problem for the adjoint linear operator which can be written as
\begin{equation*}
  E \Delta^2 \psi^{\ast} - \frac{R}{\pi^3} \Delta ( \cos(\pi y) \psi^{\ast}_x ) - \frac{R}{\pi^3} \pi^2 \cos(\pi y) \psi^{\ast}_x + \psi^{\ast}_x = \beta \Delta \psi^{\ast}.
\end{equation*}

Using $\psi^{\ast} = e^{i \alpha_m x} Y^{\ast}(y)$, we obtain the analog of \eqref{1d linear equation for Y}
\begin{equation}\label{1d adjoint linear equation for Y}
\begin{split}
 E(D^2 - \alpha_m^2)^2 Y^{\ast}& - i\alpha_m \left( \frac{R}{\pi^3} D^2 (Y^{\ast} \cos(\pi y)) + \frac{R}{\pi^3} \cos(\pi y) (\pi^2-\alpha_m^2) Y^{\ast} - Y^{\ast} \right)\\
 & = \beta^{\ast} (D^2 -\alpha_m^2)Y^{\ast}.
 \end{split}
\end{equation}
By basic properties of the adjoint linear eigenvalue problem, we have $\beta_{m,j}^{\ast} = \overline{\beta_{m, j}}$ and
\begin{equation} \label{orthogonality relation}
  ({\bf u}_{m, j}, {\bf u}^{\ast}_{n, k}) = 0, \quad \text{if } (m,j) \neq (n, k),
\end{equation}
where ${\bf u}_{m, j} = (\frac{\partial }{\partial y}, -\frac{\partial}{\partial x})\psi_{m, j}$.

\subsection{The main theorem and its proof}
Our main aim is to identify the transitions of \eqref{main} in the case where two complex
conjugate eigenvalues cross the imaginary axis. Thus, we assume the following condition
on the spectrum of the linearized operator.

\begin{assumption} \label{PES assumption}
  Depending on $a$ and $E$, there exists a critical Reynolds number $R_0$ and a zonal wave integer $m_0 \geq 1$ such that
  \begin{equation*}
  \begin{aligned}
  & \Re(\beta_{m_0, 1}(R)) = \Re(\beta_{-m_0, 1}(R)) =
  \begin{cases}
  < 0 & \text{if } R < R_0,\\
  = 0 & \text{if } R = R_0,\\
  > 0 & \text{if } R > R_0.\\
  \end{cases} \\
  & \Re (\beta_{m, j}(R_0)) < 0, \quad \text{if } (m, j) \neq (\pm m_0, 1).
  \end{aligned}
  \end{equation*}
\end{assumption}

In \cite{CGSW03}, the validity of the Assumption~\ref{PES assumption} is shown with $m_0 = 1$ when $\sqrt{3}/4 \leq a \leq \alpha_0$ for some $\sqrt{3}/4 < \alpha_0 < \sqrt{3} / 2$ .

Let us define
\begin{equation} \label{3integrals}
\begin{aligned}
& I_1 = \int \limits_{-1}^{1} \overline{Y_{m_0,1}^{\ast}}(y) \left((a m_0 \pi)^{2} - D^2 \right)Y_{m_0,1}(y) dy, \\
& I_{2}(k) =  \int \limits_{-1}^{1} \cos(k \pi y) Y_{m_0,1}^{\ast}(y) \left( (a m_0 \pi)^2 - k^2 \pi^2 - D^2 \right) \overline{Y_{m_0,1}}(y) dy, \\
& I_{3}(k) = \int \limits_{-1}^{1} \sin(k \pi y) Y_{m_0,1}(y) D\overline{Y_{m_0,1}}(y) dy.
\end{aligned}
\end{equation}
where $Y_{m_0,1}$ and $Y_{m_0,1}^{\ast}$ are solutions of \eqref{1d linear equation for Y} and \eqref{1d adjoint linear equation for Y} respectively for $m = m_0$, at $R = R_0$.

We define the transition number
\begin{equation} \label{gamma}
  \gamma = -\frac{E(R_0 a m_{0})^{2} \pi}{2 |I_{1}|^{2}}\sum_{k=1}^{\infty}\frac{\Im I_{3}(k)\Im\{I_{1}I_{2}(k)\}}{k}.
\end{equation}

As the next theorem shows, the sign of $\gamma$ determines the type of transition of the system at the critical Reynolds number $R_0$.

\begin{theorem} \label{main theorem}
Let $\gamma$  be defined by \eqref{gamma} and let
\begin{equation} \label{u_per}
  {\bf u}_{\text{bif}}(t, x, y) = \sqrt{\frac{ -\Re \beta_{m_0, 1}}{\gamma}} \Re \left( e^{i t \Im\beta_{m_0, 1}} {\bf u}_{m_0, 1}(x, y) \right) + o(\sqrt{\vert \Re \beta_{m_0, 1} \vert}).
\end{equation}
Under the Assumption~\ref{PES assumption}, the following assertions hold true:

\begin{enumerate}
\item If $\gamma < 0$ then the problem undergoes a Type-I (continuous) transition at $R = R_{0}$ and bifurcates to the time periodic solution ${\bf u}_{\text{bif}}$ on $R>R_{0}$ which is an attractor.

\item If $\gamma>0$ then the problem undergoes a Type-II (catastrophic) transition at $R = R_0$ and bifurcates to the time periodic solution ${\bf u}_{\text{bif}}$ on $R < R_{0}$ which is a repeller.
\end{enumerate}
\end{theorem}

%\subsection{Proof of the main theorem}
\begin{proof}
Let ${\bf u}_{m_0 , 1} = (\frac{\partial}{\partial y},-\frac{\partial}{\partial x})\psi_{m_0, 1}$ where $\psi_{m_0, 1} = e^{ia m_0 \pi x}Y_{m_0,1} (y)$, denote the first critical eigenfunction corresponding to the eigenvalue $\beta_{m_0, 1}$ in Assumption~\eqref{PES assumption}.
For simplicity of notation, we will denote
\begin{equation} \label{critical mode}
  \begin{aligned}
    & {\bf u}^1 = \Re {\bf u}_{m_0, 1}, \quad {\bf u}^2 = \Im {\bf u}_{m_0, 1}, \\
    & \psi^1 = \Re \psi_{m_0,1}, \quad
    \psi^2 = \Im \psi_{m_0,1}, \\
    & Y^1 = \Re Y_{m_0,1}, \quad Y^2 = \Im Y_{m_0,1},
  \end{aligned}
\end{equation}
where $Y_{m_0, 1}$ solves equation \eqref{1d linear equation for Y}.

In the proof we will use the following trilinear operators
\begin{equation} \label{G}
  G({\bf u}_I, {\bf u}_J, {\bf u}_K) =
  - \epsilon \int_{\Omega} ({\bf u}_I \cdot \nabla){\bf u}_J\cdot  \overline{{\bf u}}_Kdx dy.
\end{equation}
and
\begin{equation*}
G_s({\bf u}_I, {\bf u}_J, {\bf u}_K) = G({\bf u}_I, {\bf u}_J, {\bf u}_K) + G({\bf u}_J, {\bf u}_I, {\bf u}_K).
\end{equation*}

\noindent
\textit{Step 1. Computation of nonlinear interactions.}  It is easy to see that only $(0, k)$ and $(2m, k)$ adjoint modes interact nonlinearly with the critical $(m_0, 1)$-mode i.e.
\begin{equation} \label{reduce0}
G({\bf u}^i,{\bf u}^j,{\bf u}^{\ast}_{n,k}) = 0 \quad \text{if } n \neq 0, \text{ or } n \neq 2m, \, i,j = 1,2.
\end{equation}
We will first investigate these interactions.
We have ${\bf u}_{0k} = {\bf u}^{\ast}_{0k} = (\frac{\partial}{\partial y}\psi_{0k}, 0) = (DY_{0k}, 0)$ where $Y_{0k}$ is given by \eqref{beta0k-Y0k}. In particular ${\bf u}^{\ast}_{0k}$ is real.
Using
\begin{equation*}
G({\bf u},{\bf u}, {\bf u}_{0k}^{\ast}) = C \int_{x = 0}^{2/a} e^{2i a m_0 \pi x}dx = 0 ,
\end{equation*}
where $C = \int_{y = -1}^1 \text{(only y-dependent terms)} dy$, we get
\begin{equation} \label{NA2}
0  = G({\bf u}^1,{\bf u}^1,{\bf u}^{\ast}_{0k}) - G({\bf u}^2,{\bf u}^2,{\bf u}^{\ast}_{0k})
 + i(G({\bf u}^1,{\bf u}^2,{\bf u}^{\ast}_{0k}) + G({\bf u}^2,{\bf u}^1,{\bf u}^{\ast}_{0k}))
\end{equation}
Also noting,
\begin{equation*}
G(\overline{{\bf u}},{\bf u}, {\bf u}_{2m,k}^{\ast l})  = C \int_{x = 0}^{2/a} e^{-i\alpha_{m_0} x} e^{i\alpha_{m_0} x} T(\alpha_{2m} x) dx = 0,
\end{equation*}
where $C = \int_{y = -1}^1 \text{(only y-dependent terms)} dy$ and $T = \sin$ or $T = \cos$, we get
\begin{equation} \label{NA4}
0 = G({\bf u}^1,{\bf u}^1, {\bf u}_{2m,k}^{\ast l}) + G({\bf u}^2,{\bf u}^2, {\bf u}_{2m,k}^{\ast l})
 + i(G({\bf u}^1,{\bf u}^2, {\bf u}_{2m,k}^{\ast l}) - G({\bf u}^2,{\bf u}^1, {\bf u}_{2m,k}^{\ast l}))
\end{equation}

Let us define
\begin{equation} \label{NA5}
\begin{aligned}
& g_{0k}^{ij} = G({\bf u}^i,{\bf u}^j,  {\bf u}^{\ast}_{0k}), \\
& g_{2m,k}^{ij} = G({\bf u}^i,{\bf u}^j,{\bf u}^{\ast}_{2m,k}).\\
\end{aligned}
\end{equation}

\eqref{NA2},\eqref{NA4} and \eqref{NA5} imply that
\begin{equation} \label{NA6}
\begin{aligned}
& g_{0k}^{11}  = g_{0k}^{22} , && g_{0k}^{12} = -g_{0k}^{21}, \\
& g_{2m,k}^{11} = - g_{2m,k}^{22},  && g_{2m,k}^{12} = g_{2m,k}^{21}. \\
\end{aligned}
\end{equation}
A lengthy but straightforward calculation which can also be verified by a symbolic computation software shows that
\[
\begin{aligned}
& G({\bf u}^1,{\bf u}^1, {\bf u}_{2m,k}^{\ast 1}) = G({\bf u}^1,{\bf u}^2, {\bf u}_{2m,k}^{\ast 2}), \\
& G({\bf u}^1,{\bf u}^1, {\bf u}_{2m,k}^{\ast 2}) = -G({\bf u}^1,{\bf u}^2, {\bf u}_{2m,k}^{\ast 1}),
\end{aligned}
\]
which implies that
\begin{equation} \label{NA6.1}
g_{2m,k}^{11} = i g_{2m,k}^{12}.
\end{equation}

\noindent
\textit{Step 2. Approximation of the center manifold.}
Next we obtain an approximation for the center manifold function $\Phi$.
Let $H = E_1 \oplus E_2$, $E_1 = \text{span}\{{\bf u}^1, {\bf u}^2\}$, $E_2 = E_1^{\perp}$. By \cite{ptd}, near $R=R_0$, the center manifold $\Phi$ can be approximated by the formula
\begin{equation} \label{complex CM formula}
  \begin{split}
  ( ( -\mathcal{L})^2 + 4\Im(\beta) ^2)& ( -\mathcal{L}) \Phi  =
  ( ( -\mathcal{L}) ^2 + 2\Im(\beta) ^2) P_2G( x_1 {\bf u}^1 + x_2 {\bf u}^2) \\
  & + 2\Im(\beta) ^{2}P_{2}G(x_1{\bf u}^2-x_2{\bf u}^1)\\
  & + \Im(\beta) (-\mathcal{L}) G( x_1{\bf u}^1+x_2{\bf u} ^2,x_2{\bf u} ^1-x_1{\bf u}^2)  \\
  & + \Im(\beta) (-\mathcal{L}) G( x_2{\bf u} ^1-x_1{\bf u} ^2,x_1{\bf u} ^1+x_2{\bf u}^2) +o( 2) ,
  \end{split}
\end{equation}
where
\begin{equation*}
o( 2) = o(x_1^2 + x_2^2) +O(|\Re  \beta(R)| (x_1^2 + x_2^2)).
\end{equation*}

Here $\mathcal{L} = L\mid _{E_2}$ is the projection of the linear operator $L$ onto $E_2$.

Let us write
\[
\Phi = \sum_J \Phi_J {\bf u}_J.
\]
Note that for an eigenvector ${\bf u}_K^{\ast}$ of $L^{\ast}$ corresponding to $\beta_K^{\ast} = \overline{\beta}_K$, by orthogonality relation \eqref{orthogonality relation}, we have
\begin{equation*}
\begin{split}
( ( -\mathcal{L}) ^2+4\Im(\beta) ^2) ( -\mathcal{L}) \Phi, {\bf u}_K^{\ast}) & = \sum_K -\Phi_J\beta_K(\beta_K^2 + 4 \Im(\beta)^2) ({\bf u}_J, {\bf u}_K^{\ast}) \\
& = m_K \Phi_K,
\end{split}
\end{equation*}
where
\[
m_K = -\beta_K(\beta_K^2 + 4 \Im(\beta)^2) ({\bf u}_K, {\bf u}_K^{\ast}).
\]
By \eqref{reduce0} and the center manifold formula \eqref{complex CM formula}, we have the following approximation for the center manifold
\begin{equation} \label{complex-center-manifold}
\Phi = \sum_k \Phi _{0k}{\bf u} _{0k} + \sum_k \Phi _{2m,k}{\bf u}_{2m,k} + o(2) ,
\end{equation}
Here
\begin{equation*}
 \Phi _K = \Phi_{1,K}x_1^2+\Phi_{2,K}x_1 x_2+\Phi_{3,K}x_2^2.
\end{equation*}

Using these results, the coefficients of the center manifold function can be computed as follows:
\begin{equation} \label{NA7}
\begin{aligned}
& \Phi_{1,K} = m_K^{-1}  \left( (\beta_K^2 + 2\Im(\beta) ^2) g_K^{11} + 2\Im(\beta)^2 g_K^{22} + \Im(\beta) \beta_K (g_K^{12}+ g_K^{21}) \right), \\
& \Phi_{2,K} = m_K^{-1} \left( \beta_K^2 (g_K^{12}+g_K^{21}) + 2 \Im(\beta) \beta_K (g_K^{22}-g_K^{11}) ) \right), \\
& \Phi_{3,K} = m_K^{-1} \left( (\beta_K^2 + 2\Im(\beta) ^2) g_K^{22} + 2 \Im(\beta)^2 g_K^{11} - \Im(\beta) \beta_K (g_K^{12} + g_K^{21}) \right).
\end{aligned}
\end{equation}
Using \eqref{NA6}, \eqref{NA6.1} and \eqref{NA7}, we find that
\begin{equation*}
\begin{aligned}
& \Phi_{1,0k} = \Phi_{3,0k}, \quad  \Phi_{2,0k} =  0, \\
& \Phi_{1,2m,k} = \frac{i}{2} \Phi_{2,2m,k} = - \Phi_{3,2m,k}
\end{aligned}
\end{equation*}
Hence the center manifold function \eqref{complex-center-manifold} becomes
\begin{equation} \label{complex-center-manifold2}
\Phi  = \sum_k \Phi_{1,0k} {\bf u} _{0k} (x_1^2 + x_2^2) + \Phi_{1,2m,k} ( x_1^2 -2 i  x_1 x_2 - x_2^2) {\bf u}_{2m,k} +o(2).
\end{equation}
Note that $\Phi_{1,0k}$ is real while $\Phi_{1,2m,k}$ is complex.

\medskip

\noindent
\textit{Step 3. Construction of adjoint modes.}  Now we construct the adjoint modes ${\bf U}^{\ast 1}$ and ${\bf U}^{\ast 2}$ orthogonal to ${\bf u}^1$ and ${\bf u}^2$. Let us denote the real and imaginary parts of the critical adjoint eigenvector by ${\bf u}^{\ast 1} = \Re {\bf u}^{\ast}_{m_0,1}$ and ${\bf u}^{\ast 2} = \Im {\bf u}^{\ast}_{m_0,1}$ and define
\begin{equation} \label{Uast}
  {\bf U}^{\ast 1} = \frac{({\bf u}^1, {\bf u}^{\ast 1}) {\bf u}^{\ast 1} + ({\bf u}^1, {\bf u}^{\ast 2}) {\bf u}^{\ast 2}}{({\bf u}^1, {\bf u}^{\ast 1})^2 + ({\bf u}^1, {\bf u}^{\ast 2})^2}, \quad
  {\bf U}^{\ast 2} = \frac{-({\bf u}^1, {\bf u}^{\ast 2}) {\bf u}^{\ast 1} + ({\bf u}^1, {\bf u}^{\ast 1}) {\bf u}^{\ast 2}}{({\bf u}^1, {\bf u}^{\ast 1})^2 + ({\bf u}^1, {\bf u}^{\ast 2})^2}.
\end{equation}
Noting that for any two functions of the form $f_i(x, y) = e^{i a m \pi x} g_i(y)$ for $i=1,2$ where $m$ is a nonzero integer, we have
\[
   (\Re f_1, \Re f_2) = (\Im f_1, \Im f_2), \quad (\Re f_1, \Im f_2) = - (\Im f_1, \Re f_2).
\]
Thus
\[
({\bf u}^1, {\bf u}^{\ast 1}) = ({\bf u}^2, {\bf u}^{\ast 2}), \quad
({\bf u}^1, {\bf u}^{\ast 2}) = -({\bf u}^2, {\bf u}^{\ast 1}).
\]
which implies that
\begin{equation} \label{(u, Uast)}
  \begin{aligned}
    & ({\bf u}^2, {\bf U}^{\ast 1}) = ({\bf u}^1, {\bf U}^{\ast 2}) = 0, \\
    & ({\bf u}^1, {\bf U}^{\ast 1}) = ({\bf u}^2, {\bf U}^{\ast 2}) = 1.
  \end{aligned}
\end{equation}

\medskip

\noindent
\textit{Step 4. Derivation of the reduced equations.}
Now we write
\begin{equation} \label{u_CM}
  {\bf u}(x, y, t) = x_1(t) {\bf u}^1(x, y) + x_2(t) {\bf u}^2(x, y) + \Phi(x, y, t)
\end{equation}
where $\Phi$ is the center manifold function, $x_1,x_2\in \mathbb{R}$. Note that
\[
  L ({\bf u}^1 + i {\bf u}^2 ) = \beta ({\bf u}^1 + i {\bf u}^2)
\]
implies
\[
  L {\bf u}^1 = \Re(\beta) {\bf u}^1 - \Im(\beta) {\bf u}^2, \qquad
  L {\bf u}^2 = \Re(\beta) {\bf u}^2 + \Im(\beta) {\bf u}^1.
\]
Also by definition of center manifold, we have $(\Phi, {\bf u}^{\ast 1}) = 0$ and $(\Phi, {\bf u}^{\ast 2}) = 0$ which by \eqref{Uast} implies that
\begin{equation} \label{orthogonality of cm to critical eig}
  (\Phi, {\bf U}^{\ast i}) = 0, \qquad i = 1, 2.
\end{equation}
Hence by \eqref{(u, Uast)} and \eqref{orthogonality of cm to critical eig}
\begin{equation} \label{(Lu,Uast)}
  \begin{aligned}
    & (L {\bf u}, {\bf U}^{\ast 1}) = (x_1 L {\bf u}^1, {\bf U}^{\ast 1}) + (x_2 L {\bf u}^2, {\bf U}^{\ast 1}) = \Re(\beta) x_1 + \Im(\beta) x_2, \\
    & (L {\bf u}, {\bf U}^{\ast 2}) = (x_1 L {\bf u}^1, {\bf U}^{\ast 2}) + (x_2 L {\bf u}^2, {\bf U}^{\ast 2}) = \Im(\beta) x_1 + \Re(\beta) x_2.
  \end{aligned}
\end{equation}
Plugging \eqref{u_CM} into \eqref{main functional equation},  taking inner product with ${\bf U}^{\ast i}$ ($i=1,2$) and using \eqref{(u, Uast)} and \eqref{(Lu,Uast)}, we can write the reduced equations as
\begin{equation}
\begin{aligned} \label{real-reduced}
& \frac{d x_1}{dt} = \Re(\beta) x_1 + \Im(\beta) x_2 + (G( {\bf u} ,{\bf u} ) , {\bf U}^{\ast 1}) ,  \\
& \frac{d x_2}{dt} = -\Im(\beta) x_1+ \Re(\beta) x_2 + (G( {\bf u} ,{\bf u} ) , {\bf U}^{\ast 2}).
\end{aligned}
\end{equation}
Noting
\[
G_s({\bf u}^i, {\bf u}^j, {\bf U}^{\ast k}) = 0, \quad i,j,k =1,2,
\]
and
\[
  G_s(\Phi, \Phi, {\bf U}^{\ast k}) = o(3).
\]
we can expand the nonlinear terms of $x_1$ and $x_2$ in \eqref{real-reduced}
\begin{equation} \label{comp1}
( G( {\bf u} ,{\bf u} ) , {\bf U}^{\ast j}) =
x_1 G_s({\bf u}^1,\Phi, {\bf U}^{\ast j}) + x_2 G_s({\bf u}^2,\Phi,  {\bf U}^{\ast j}) + o(3).
\end{equation}
By \eqref{complex-center-manifold2},
\begin{equation} \label{comp2}
G_s({\bf u}^i,\Phi, {\bf U}^{\ast j}) = \sum_k  \Phi_{1,0k} (x_1^2 + x_2^2) c_{0k}^{ij} + \Phi_{1,2m,k}(x_1^2 -2i x_1 x_2 - x_2^2) c_{2m,k}^{ij},
\end{equation}
where we define
\begin{equation} \label{comp3}
\begin{aligned}
& c_{0,k}^{ij} = G_s({\bf u}^i,{\bf u}_{0k}, {\bf U}^{\ast j}), \\
& c_{2m,k}^{ij} = G_s({\bf u}^i, {\bf u}_{2m,k}, {\bf U}^{\ast j}).
\end{aligned}
\end{equation}
As in \eqref{NA5}, we can show that
\begin{equation} \label{comp4}
\begin{aligned}
& c_{0k}^{11} = c_{0k}^{22}, && c_{0k}^{12} = -c_{0k}^{21}, \\
& c_{2m,k}^{11} = - c_{2m,k}^{22}, && c_{2m,k}^{12} = c_{2m,k}^{21}.
\end{aligned}
\end{equation}
Moreover the calculation
\begin{equation} \label{comp5}
\begin{aligned}
& G_s({\bf u}^1, {\bf u}_{2m}^1, {\bf U}^{\ast 1}) = G_s({\bf u}^1, {\bf u}_{2m}^2, {\bf U}^{\ast 2}), \\
& G_s({\bf u}^1, {\bf u}_{2m}^2, {\bf U}^{\ast 1}) = -G_s({\bf u}^1, {\bf u}_{2m}^1, {\bf U}^{\ast 2})
\end{aligned}
\end{equation}
implies that
\begin{equation} \label{c2mk11}
c_{2m,k}^{11} = -i c_{2m,k}^{12}
\end{equation}
Using \eqref{comp1}-\eqref{c2mk11}, we have
\begin{equation*}
( G( {\bf u} ,{\bf u} ) , {\bf U} ^{\ast j}) =
b_{30}^j x_1^{3} + b_{21}^j x_1^2 x_2 + b_{12}^j x_1 x_2^2 + b_{03}^j x_2^{3} + o(3), \quad j =1,2,
\end{equation*}
where simple calculations show that
\begin{equation} \label{b130 b230}
\begin{aligned}
& b_{30}^1 = b_{12}^1 = b_{21}^2 = b_{03}^2 =
\sum_k \Phi_{1,0k} c_{0k}^{11} + \Phi_{1,2m,k} c_{2m,k}^{11}, \\
& b_{30}^2 = b_{12}^2 = - b_{21}^1 = - b_{03}^1 =
\sum_k \Phi_{1,0k} c_{0k}^{12} + i \Phi_{1,2m,k} c_{2m,k}^{11}.
\end{aligned}
\end{equation}
Thus the reduced equations \eqref{real-reduced} become
\begin{equation}\label{real-reduced-2}
\begin{aligned}
&\frac{d x_1}{dt} =\Re(\beta) x_1+\Im(\beta)
x_2 + b_{30}^1 x_1(x_1^2 + x_2^2) - b_{30}^2 x_2(x_1^2 + x_2^2) +o(3) ,\\
&\frac{dx_2}{dt} =-\Im(\beta) x_1+\Re(\beta)
x_2 + b_{30}^2 x_1(x_1^2 + x_2^2) + b_{30}^1 x_2(x_1^2 + x_2^2) +o(3).
\end{aligned}
\end{equation}

\medskip

\noindent
\textit{Step 5. Computation of the transition number $\gamma$.}
Letting $z = x_1+i x_2$, \eqref{real-reduced-2} becomes
\begin{equation}\label{complexreduced}
\frac{dz}{dt} = \overline{\beta} z + b z|z|^2 +o(|z|^3),
\end{equation}
where by \eqref{c2mk11} and \eqref{b130 b230}
\begin{equation} \label{b}
b = b_{30}^1+i b_{30}^2 = \sum_k \Phi_{1,0k} (c_{0k}^{11} + i c_{0k}^{12})
\end{equation}

If $\Re b < 0$ then \eqref{complexreduced} has a stable limit cycle
\[
   z = r e^{i \omega t}
\]
for $\Re \beta > 0$, i.e. for $R > R_0$ with
\[
  r = \sqrt{-\frac{\Re \beta}{\Re b}}, \quad \omega = -\Im \beta + \Im b r^2 \approx -\Im \beta.
\]

If $\Re b > 0$ then \eqref{complexreduced} has an unstable limit cycle for $R < R_0$.

% Note:
% If we multiply \eqref{complexreduced} by $\overline{z}$, we get
% \begin{equation*}
% \frac{1}{2}\frac{d}{dt} |z|^2 = \overline{\beta} |z|^2 + b |z|^4 +o(|z|^4)
% \end{equation*}

Thus the transition is determined by the sign of real part of $b$ in \eqref{b} at $R=R_0$ defined as
\begin{equation} \label{gamma pre}
  \gamma = \Re(b) = \sum_{k=1}^{\infty} \Phi_{1,0k} c_{0k}^{11}
\end{equation}
since $\Phi_{1,0k}$, $c_{0k}^{11}$, $c_{0k}^{12}$ are real numbers.

\textit{Step 6. Derivation of the transition number $\gamma$ in \eqref{gamma}.}

Using \eqref{NA6} in \eqref{NA7}, we get
\begin{equation} \label{Phi10k}
  \Phi_{1,0k} = \frac{g_{0k}^{11}}{-\beta_{0k}({\bf u}_{0k},{\bf u}_{0k}^{\ast})}
\end{equation}
The mode ${\bf u}_{0k} = {\bf u}^{\ast}_{0k} = (\frac{\partial}{\partial y}\psi_{0k}, 0) = (DY_{0k}, 0)$ corresponds to the eigenfunction with $m = 0$ in \eqref{1d linear equation for Y}, i.e. solutions of $E D^4 Y = \beta D^2 Y$ with the boundary conditions $Y(\pm 1) = D^2 Y(\pm 1) = 0$. These solutions are easily obtainable.
\begin{equation} \label{beta0k-Y0k}
  \begin{aligned}
      & \beta_{0k} = - \frac{k^2 \pi^2}{4} E, \\
      & Y_{0k} =
      \begin{cases}
        \cos \frac{k \pi y}{2}, & \text{if k is odd}, \\
        \sin \frac{k \pi y}{2}, & \text{if k is even}. \\
      \end{cases}
  \end{aligned}
\end{equation}
Thus
\[
  ({\bf u}_{0k},{\bf u}_{0k}^{\ast})=\frac{2}{a}\int_{-1}^{1}\vert DY_{0k}\vert^{2} = \frac{-2\beta_{0k}}{aE}.
\]
To obtain $\gamma$, we need to compute
\begin{equation} \label{c0k11}
  c_{0,k}^{11} = G_s({\bf u}^1,{\bf u}_{0k}, {\bf U}^{\ast 1}) =
  G({\bf u}^1,{\bf u}_{0k}, {\bf U}^{\ast 1}) + G({\bf u}_{0k}, {\bf u}^1, {\bf U}^{\ast 1})
\end{equation}
and
\begin{equation} \label{g0k11}
  g_{0k}^{11} = G({\bf u}^1,{\bf u}^1,  {\bf u}^{\ast}_{0k})
\end{equation}
where ${\bf u}^1$ is the real part of the critical eigenfunction (with zonal wave number $m_0$) as given in \eqref{critical mode}, $U^{\ast 1}$ is given by \eqref{Uast} and $G$ is the trilinear operator \eqref{G}. So we plug in $\psi_{m_0 1} = e^{i a m_0 \pi x} Y_{m_01}(y)$ and $\psi^{\ast}_{m_0 1} = e^{i a m_0 \pi x} Y_{m_01}^{\ast}(y)$ into \eqref{c0k11} and \eqref{g0k11}. After tedious computations, we can obtain
\begin{equation} \label{c0k11-g0k11 final}
  \begin{aligned}
    & c_{0k}^{11}  = - \frac{a\epsilon m_0 \pi}{|J_1|^{2}} \Im\{J_1 J_{2}(k)\}, \\
    & g_{0k}^{11} = -\frac{\epsilon m_0 \pi \beta_{0k}}{E} \Im J_{3}(k),
  \end{aligned}
\end{equation}
where
\begin{equation} \label{3integrals-ver1}
\begin{aligned}
& J_1 = \int \limits_{-1}^{1} \overline{Y^{\ast}} ((a m_0 \pi)^{2} - D^2)Y dy, \\
& J_{2}(k) =  \int \limits_{-1}^{1} D Y_{0k} \overline{Y^{\ast}} \left( (a m_0 \pi)^2 - \frac{k^2 \pi^2}{4} - D^2 \right)Y dy, \\
& J_{3}(k) = \int \limits_{-1}^{1} Y_{0k} Y D\overline{Y} dy.
\end{aligned}
\end{equation}
Finally we note that $Y = Y_{m_01}$ is either an odd or an even function of $y$ which follows from the fact that the equation \eqref{1d linear equation for Y} is invariant under the change of variable $y \rightarrow -y$. Thus $Y_{0k}$ must be odd (by \eqref{beta0k-Y0k}, $k$ must be even) otherwise $J_3(k)$ is zero. So the nonzero contributions to $\gamma$ in \eqref{gamma pre} comes from even $k$. Thus we obtain \eqref{gamma} from \eqref{gamma pre}, \eqref{Phi10k}, \eqref{beta0k-Y0k}, \eqref{c0k11-g0k11 final} and \eqref{3integrals-ver1}, which concludes the proof.
\end{proof}

\section{Numerical Evaluation}

\subsection{Legendre-Galerkin method for \eqref{1d linear equation for Y} and \eqref{1d adjoint linear equation for Y}}

In this section we present a method to approximate the solutions of the eigenvalue problems \eqref{1d linear equation for Y} and \eqref{1d adjoint linear equation for Y}. There are  fourth-order problems so a Legendre-Galerkin method \cite{Shen94,STW11} will be efficient and accurate.

We look for an approximation of $Y$ in the space
$X_N=\{v\in P_N: v(\pm 1 ) = v''(\pm 1) = 0\}$, where $P_N$ is the space of polynomials with degree less than or equal to $N$.

Using the approach in \cite{Shen94}, we set
\[
  f_j(y) = L_j(y)+\sum_{k = 1}^{4} c_{j, k} L_{j+k}(y),
\]
with $c_{j, k}$ to be chosen such that
\[
  f_j(\pm 1 ) = f''_j(\pm 1) = 0.
\]
It is easy to determine from the properties of Legendre polynomials that
\[
\begin{aligned}
  & c_{j, 1} = c_{j, 3} = 0, \\
  & c_{j, 2} = \frac{2(2j+5)(j^2+5j+9)}{(j+3)(j+4)(2j+7)}, \\
  & c_{j, 4} = -1 - c_{j, 2},
\end{aligned}
\]
and we have $X_N=\text{span}\{f_j: j=0,1,\cdots,N-4\}$.

Writing  $Y^N(y) = \sum_{j = 0}^{N-4} y_j f_j(y)\in X_N$, and plugging it into \eqref{1d linear equation for Y}, and taking inner product with $f_k(y)$, ($k = 0, 1, \dots, N-4$) we obtain the Legendre-Galerkin approximation of \eqref{1d linear equation for Y} in the following matrix form:
\begin{equation} \label{disc_eig_prob}
\begin{split}
  & \left( E (A_1 - 2\alpha_m^2 A_2+ \alpha_m^4 A_3) + i \alpha_m ( \frac{R}{\pi^3} A_4^T + \frac{R}{\pi^3}(\pi^2-\alpha_m^2)A_5 - A_3 ) \right)\mathcal{Y}^N \\
  & = \beta (A_2 - \alpha_m^2 A_3) \mathcal{Y}^N,
\end{split}
\end{equation}
where
\begin{equation*}
\begin{aligned}
& a_{1, jk} = (D^4f_j,f_k), \quad a_{2, jk} = (D^2f_j,f_k), \quad a_{3, jk} = (f_j,f_k), \\
& a_{4, jk} = (\cos \pi y \, D^2f_j,f_k), \quad a_{5, jk} = (\cos \pi y \, f_j,f_k) \\
& A_i = (a_{i, jk})_{j, k = 0, \dots N-4} \quad i = 1, \dots 5\\
& \mathcal{Y}^N =
\begin{bmatrix}
  y_0 && y_1 && \cdots && y_{N-4}
\end{bmatrix}^T.
\end{aligned}
\end{equation*}
By using the properties of Legendre polynomials, we find that $A_1$, $A_2$, $A_3$ are real symmetric banded matrices given by
\[
  a_{1,jk} =
  \begin{cases}
    \dfrac{(2 + 2j)(2 + j)(3 + 2j)^2(5 + 2j)}{(3 + j)(4 + j)}, &\text{ if }j = k\\
    0, &\text{ otherwise }
  \end{cases}
\]
\[
  a_{2,jk} = \begin{cases}
    \dfrac{(2 + 2j)(2 + j)(3 + 2j)}{(3 + j)(4 + j)}, &\text{ if } j = k \pm 2 \\
    \dfrac{(4(3 + 2j)(5 + 2j)(102 + 110j + 47j^2 + 10j^3
    +j^4)}{(3 + j)^2(4 + j)^2(7 + 2j)}, &\text{ if } j = k\\
    0, &\text{ otherwise }
  \end{cases}
\]
\[
  a_{3,jk} = \begin{cases}
    \dfrac{2(1 + j)(2 + j)(3 + 2j)}{(3 + j)(4 + j)(7 + 2j)(9 + 2j)}, &\text{ if } j = k \pm 4 \\
    \dfrac{-8(222 + 196j + 77j^2 + 14j^3 + j^4)}{(3 + j)(4 + j)(5 +j)(6 +j)(11 + 2j)}, &\text{ if } j = k \pm 2 \\
    \dfrac{4(5580 + 11202j + 9263j^2 + 4170j^3 + 1105j^4
    +168j^5 + 12j^6)}{(3 + j)^2(4 + j)^2(1 + 2j)(7 + 2j)(9 + 2j)}, &\text{ if } j = k\\
    0, &\text{ otherwise }
  \end{cases}
\]
To approximate $A_4$ and $A_5$ we fix some integer $M$ and compute the Legendre-Gauss-Lobatto quadrature points $y_n$  and weights $\omega_n$ for $n = 0, \dots, M$. Then we compute the matrices $D^2 f_j(y_n)$ and $f_j(y_n)$.
\begin{align*}
  & a_{4, jk} = \sum_{n=0}^M \cos(\pi y) D^2 f_j(y_n) f_k(y_n) \omega_n, \\
  & a_{5, jk} = \sum_{n=0}^M \cos(\pi y) f_j(y_n) f_k(y_n) \omega_n.
\end{align*}
$M$ is chosen large enough to provide sufficient accuracy in the computation of $A_4$ and $A_5$. We note that $A_4$ and $A_5$ are real, full matrices. $A_5$ is symmetric while $A_4$ is non-symmetric.

The eigenvalue problem \eqref{disc_eig_prob} can be easily solved by using a standard eigenvalue solver.
For each $m \in \mathbb{Z}$, we can numerically find $N-3$ eigenvalues $\beta^N_{m, j}$ (with $\Re \beta^N_{m, j} \ge \Re \beta^N_{m, k}$ if $j < k$) of \eqref{disc_eig_prob}, and corresponding eigenvectors $\mathcal{Y}^N_{m, j}$, $j = 1, \dots, N-3$.

By taking the complex conjugate of \eqref{disc_eig_prob}, we find $\mathcal{Y}^N_{m, j} = \overline{\mathcal{Y}}^N_{-m, j}$, $\beta^N_{m, j} = \overline{\beta}^N_{-m, j}$ and in particular $\beta^N_{0j}$ and $\psi^N_{0j}$ are real.

It is known that (\cite{Wei.T88}), the computed eigenpairs $(\beta^N_{m, j}, e^{i\alpha_m x}Y^N_{m, j})$ of \eqref{disc_eig_prob} converge to eigenpairs $(\beta_{m,j},\, \psi_{m,j})$ of \eqref{lin_eig_prob} exponentially as $N\rightarrow \infty$  for $0\le j\lesssim 2N/\pi$ for each fixed $m$.

Finally, the analog of approximating equation \eqref{disc_eig_prob} for the adjoint problem \eqref{1d adjoint linear equation for Y} is the adjoint problem of \eqref{disc_eig_prob}. Namely
\begin{equation} \label{adjoint disc_eig_prob}
\begin{split}
  & \left( E (A_1 - 2\alpha_m^2 A_2+ \alpha_m^4 A_3) - i \alpha_m ( \frac{R}{\pi^3} A_4 + \frac{R}{\pi^3}(\pi^2-\alpha_m^2)A_5 - A_3 ) \right)\mathcal{Y}^{\ast N} \\
  & = \beta^{\ast} (A_2 - \alpha_m^2 A_3) \mathcal{Y}^{\ast N}.
\end{split}
\end{equation}

\subsection{Numerical computation of the transition number $\gamma$} \label{Numerical approximation of the transition number}

To approximate $\gamma$ in \eqref{gamma}, we follow the following steps:
\begin{enumerate}[label=\textbf{Step \arabic*}]
  \item The parameters of the system are $a$, $\epsilon$, $E$ and $R$. We fix the Ekman number $E$ and the length scale $a$. Then by \eqref{RepsilonE}, the Rossby number $\epsilon = R / E $ is also fixed.
  Moreover, to study the transition, $R$ has to be fixed to $R = R_0$ where $R_0$ is determined by Assumption~\ref{PES assumption}.

  \item

  \begin{figure}[ht]
  \includegraphics[scale=1]{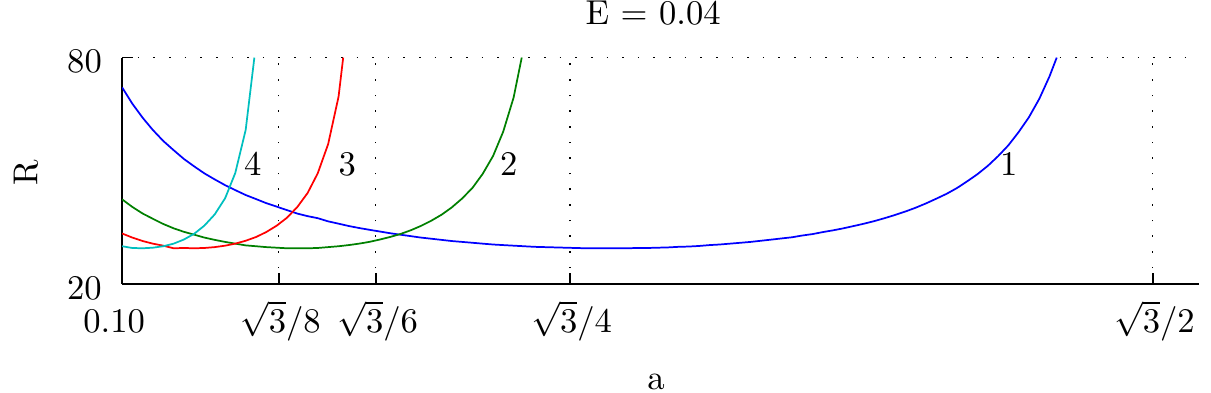}
  \includegraphics[scale=1]{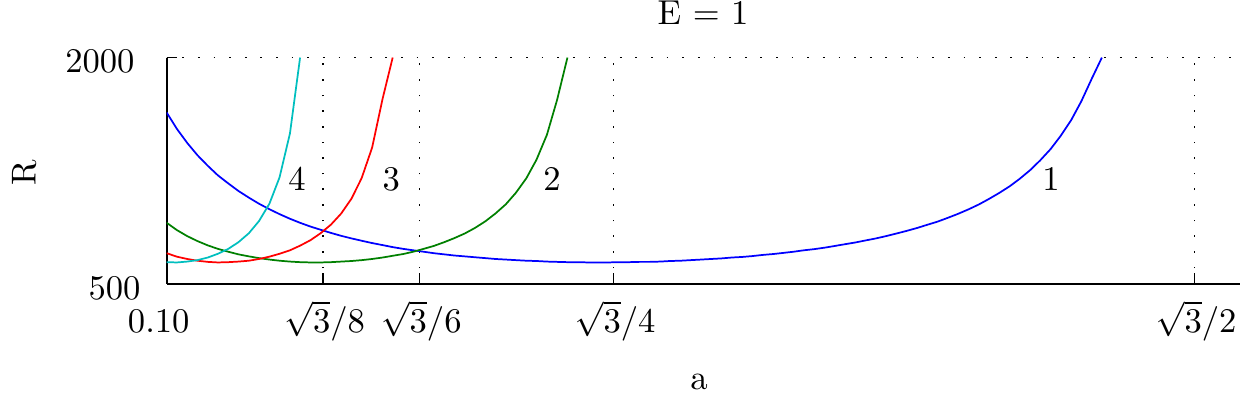}
  \caption{Neutral stability curves which are defined by $\Re\beta_{m, 1} = 0$ in the a-R plane for zonal wave numbers $m=1,2,3,4$ and for $E=0.04$, $E=1$.}
  \label{fig:neut_stab}
  \end{figure}

  We determine the critical zonal wave number $m_0$ and the critical Reynolds number $R_0$ in Assumption~\ref{PES assumption}.
  In \cite{CGSW03}, it was shown that $m_0 = 1$ if $\sqrt{3}/4 \le a \le \alpha_0$ for some $\alpha_0 <\sqrt{3}/2$.
  Figure~\ref{fig:neut_stab} show that the difficulties in proving the Assumption~\ref{PES assumption} for $\alpha_0 \le a < \sqrt{3}/2$ is purely technical.
  Extrapolating this result supported by our numerical computations of neutral stability curves shown in Figure~\ref{fig:neut_stab}, we claim that Assumption~\ref{PES assumption} can be satisfied by $m$ only if $a < \sqrt{3}/(2 m)$.

  For such $m$, we set $\Re \beta^N_{m, 1} = 0$ in \eqref{disc_eig_prob} and solve for the Reynolds number $R$  to find $R^N_{0, m}$. $R^N_0$ which approximates $R_0$, is the minimum of such $R^N_{0, m}$ and the minimizing $m$ value is $m_0$.

  Figure~\ref{fig:neut_stab} shows that Assumption~\ref{PES assumption} holds for a simple complex pair of eigenvalues for almost all $a < \sqrt{3}/2$ except for discrete values of $a$ where neutral stability curves corresponding to different zonal wave numbers intersect; Assumption 1 is generic.

  \item With $\beta^N_{m_0, 1}$ computed, the eigensolutions $\mathcal{Y}^N_{m_0, 1} = [\tilde{y}_j]_{j=0}^{N-4}$ and $\mathcal{Y}^{\ast N}_{m_0, 1} = [\tilde{y}^{\ast}_j]_{j=0}^{N-4}$ of \eqref{disc_eig_prob} and \eqref{adjoint disc_eig_prob} are also found.

  \item The Legendre-Gauss-Lobatto quadrature points $y_j$  and weights $\omega_j$ are calculated for $j = 0, \dots, M$ where $M$ will be chosen large enough to allow sufficient accuracy in the computation of \eqref{disc 3 integrals}.
  \item The values of the eigenmodes and their derivatives at quadrature points $y_j$ are computed.
  \[
    \begin{aligned}
      & Y_{m_0, 1}^N(y_j) = \sum_{k=0}^{N-4} \tilde{y}_k f_k(y_j), \qquad Y_{m_0, 1}^{\ast N}(y_j) = \sum_{k=0}^{N-4} \tilde{y}^{\ast}_k f_k(y_j), \\
      & DY_{m_0, 1}^N(y_j) = \sum_{k=0}^{N-4} \tilde{y}_k f_k'(y_j), \qquad D^2Y_{m_0, 1}^N(y_j) = \sum_{k=0}^{N-4} \tilde{y}_k f_k''(y_j). \\
    \end{aligned}
  \]
  \item It is easy to see that if we multiply $Y_{m_0, 1}$ by a complex number $c$ then $\gamma$ is multiplied by $\vert c \vert^2$. To find a unique $\gamma$, we normalize $Y_{m_0, 1}$ so that $\max_{0\le j \le M} Y^N_{m_0, 1}(y_j) = 1$.
  \item We can also normalize $Y^{\ast N}_{m_0, 1}$ so that $I_1 = 1$ in \eqref{3integrals}. It is easy to see that this scaling has no effect on $\gamma$.
  \item Finally we approximate $I_2(k)$, $I_3(k)$ by
  \begin{equation} \label{disc 3 integrals}
    \begin{aligned}
      & I_{2}(k) \approx  \sum_{j=0}^M \cos(k \pi y_j) Y_{m_0,1}^{\ast N}(y_j) \left( (a m_0 \pi)^2 - k^2 \pi^2 - D^2 \right) \overline{Y^N_{m_0,1}(y_j)} \omega_j, \\
      & I_3(k) \approx \sum_{j=0}^M \sin(k \pi y_j) Y^N_{m_0,1}(y_j) \overline{D Y_{m_0,1}^N(y_j)} \omega_j.
    \end{aligned}
  \end{equation}
  Obviously increasing $M$ increases the accuracy of approximation in \eqref{disc 3 integrals}. In our experiments we found that $M = 2N$ gives enough accuracy.
\end{enumerate}

\subsection{Numerical results}  \label{Section: numerical results}

In this section, we present the results of our numerical computations of $\gamma$ for
different parameter choices. As discussed in Section~\ref{Numerical approximation of the
transition number}, the only parameters that need to be varied are the Ekman number $E$
and the length scale $a$.

For showing typical results, we consider a mid-latitude atmospheric jet in a zonal channel
at a reference latitude $\theta_0 = 45^{\circ}$N. The dimensional zonal velocity of the
background state $u_0$ in \eqref{steady state} has a maximum $U/(\pi^3 E)$.  With a typical
zonal velocity of $U = 15$ ms$^{-1}$, we limit our numerical investigations to $E$ values
between $0.01$ to $0.03$ which corresponds to maximum zonal jet velocities between
$16$~ms$^{-1}$ up to $48$~ms$^{-1}$.  For $a=0.2$, a typical length scale of $L = 3000$~km
yields a channel of length $2 L/a = 30000$~km in the meridional direction that corresponds to
about $360^{\circ}$ in longitude. We therefore consider $a$ values in the range $0.1 \le a \le
0.6$.

For this parameter regime we approximate the transition number $\gamma$ following the
procedure in Section~\ref{Numerical approximation of the transition number}. The results
we find are presented in Table~\ref{results-table} which suggest that $\gamma$ is always
positive. According to Theorem~\ref{main theorem}, this means the transition is Type-II at
the first critical Reynolds number $R_0$.

We can also compute the period $T = 2\pi/\Im\beta_{m_0, 1}$ of the solution
\eqref{u_per} where $T$ has been non-dimensionalized by $1/(\beta_0 L)$. For example,
with the above  choices of $L$ and $U$, for $a = 0.20$ and $E=0.01$ we find that
$m_0 = 2$, $R_0 = 5.64$. The planetary vorticity gradient $\beta_0$ at $45$N is $1.6\times 10^{-11}$. We compute the period to be about $T=196$ days. The
stream function of the time periodic  solution \eqref{u_per} bifurcated on $R<R_0$
which is unstable because of the Type-II  transition is shown in
Figure~\ref{fig:periodic_solution}.  The pattern indicates a typical one due to
barotropic instability  which, due to the background state zonal velocity, propagates
eastwards.

\begin{figure}[ht]
\begin{center}
$${}\hskip -.8in\includegraphics[scale=1]{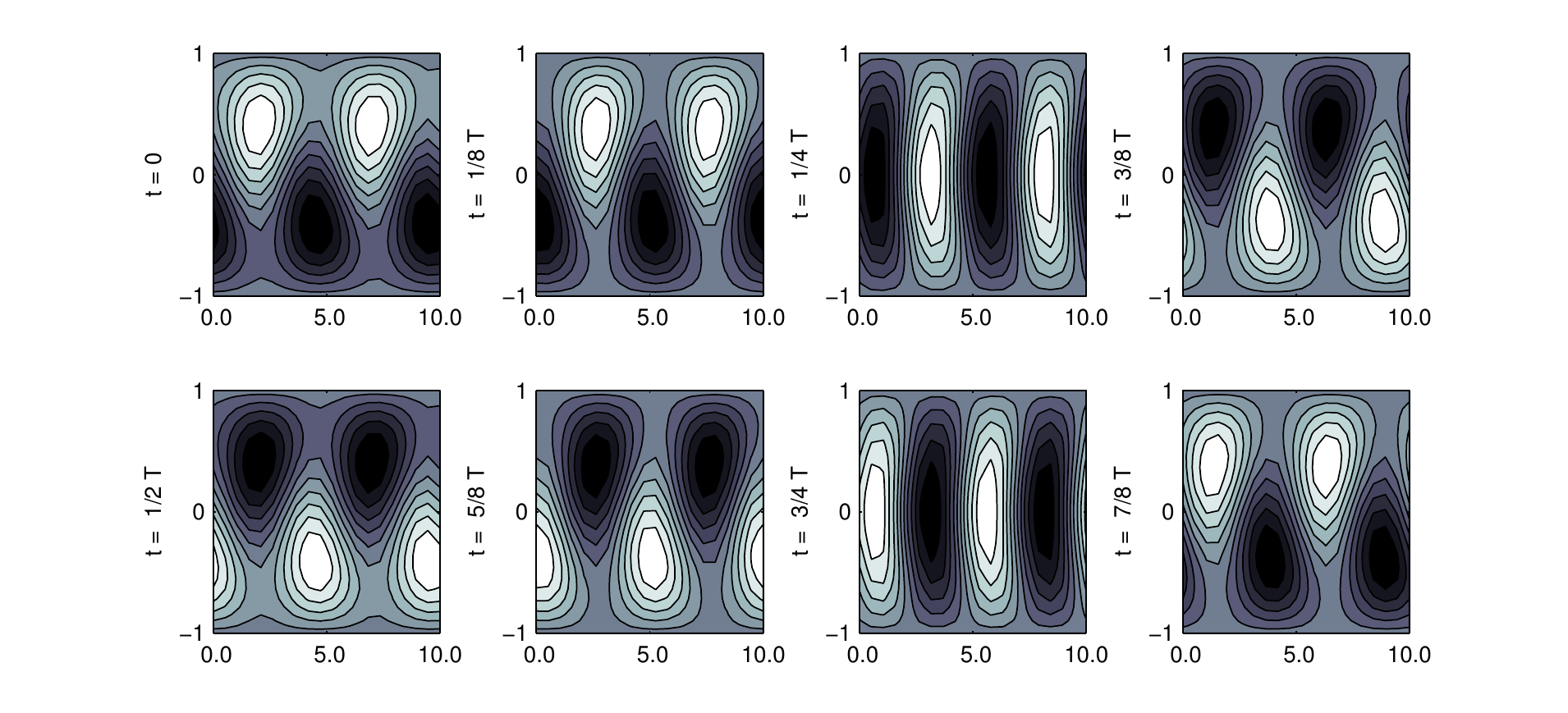}$$
\end{center}
\caption{The stream function of the bifurcated periodic solution on $R<R_0$ at
$a=0.2$ and $E=0.01$;  $T$ denotes the  period of the periodic orbit.}
\label{fig:periodic_solution}
\end{figure}

\begin{table}
\caption{The value of $\gamma$ for $0.1 \le a \le 0.6$ and $0.01 \le E \le 0.03$.
\label{results-table}}
\begin{tabular}{cc|cccccc|}
\cline{3-8}
 & & \multicolumn{6}{|c|}{a} \\
\cline{3-8}
&  & 0.100 & 0.200 & 0.300 & 0.400 & 0.500 & 0.600 \\
\hline
\multicolumn{1}{|c|}{\multirow{5}{*}{E}} &
0.010 & 0.197 & 0.197 & 0.188 & 0.197 & 0.194 & 0.175 \\
\multicolumn{1}{|c|}{}& 0.015 & 0.174 & 0.174 & 0.163 & 0.174 & 0.172 & 0.155 \\
\multicolumn{1}{|c|}{}& 0.020 & 0.156 & 0.156 & 0.143 & 0.156 & 0.156 & 0.143 \\
\multicolumn{1}{|c|}{}& 0.025 & 0.144 & 0.144 & 0.130 & 0.144 & 0.146 & 0.135 \\
\multicolumn{1}{|c|}{}& 0.030 & 0.135 & 0.135 & 0.121 & 0.135 & 0.139 & 0.130 \\
\hline
\end{tabular}

\end{table}

\section{Summary and Discussion}

In this paper, we have extended the results of  \cite{CGSW03} on the existence of a Hopf
bifurcation in the forced barotropic vorticity equation (\ref{pre-main}) by rigorously showing
the type  of finite amplitude solutions which can occur near this Hopf bifurcation. Central in the
analysis is the computable quantity $\gamma$ which characterizes the transition behavior
near the critical point.

The aim of this paper was only to focus on the theory and numerical computation. As an
illustration, we explore this number numerically in a parameter  regime relevant for an
atmospheric jet and find that a catastrophic transition is preferred.  For typical  ocean cases,
for example in western boundary currents such as the Gulf Stream  and the Kuroshio, and
the Antarctic Circumpolar Current, the  results will be reported  elsewhere.

While (\ref{pre-main}) is a cornerstone dynamical model of the ocean and atmospheric circulation,
it of course represents only a limited number of processes.  As practiced by the earlier workers in this
field, such as J. Charney  and J. von Neumann, and from the lessons learned by the failure of Richardson's pioneering work, one tries to be satisfied with simplified models approximating the actual motions to a greater or lesser degree instead of attempting to deal with the atmosphere/ocean in all its complexity. By starting with models incorporating only what are thought to be the most important of atmospheric influences, and by gradually bringing in others, one is able to proceed inductively and
thereby to avoid the pitfalls inevitably encountered when a great many poorly understood factors are introduced all at once.

However, the same results will be true if we work on the barotropic  equations in
primitive  variables; see (\ref{primitive-formulation}) and \cite{SD13}.  Second, we
would expect  that the same type of results obtained in this article will also
be true if we use higher order friction (e.g. hyper-friction,  \cite{selton}) in
(\ref{pre-main}) as this would only change the eigenfunction structure slightly.
Third, the method presented in this paper combined with the methods introduced in
\cite{CLW14a, CLW14b}  can be used to study the case with bottom topography
\cite{COV04}, and the case where  noise represents  ocean eddies. We will
explore these new directions in the near future.
\medskip

\noindent{Acknowledgment.} The work of H.D. was supported by the Netherlands Organization for  Scientific Research (NWO) through the  COMPLEXITY  project PreKurs; the work of T.S. was supported by the Scientific and Technological Research Council of Turkey (Grant number 114C142);  the work of J.S. was partially supported by NSF grants DMS-1217066 and DMS-1419053; and the work of S.W. was supported in part by NSF grant DMS-1211218.

\bibliographystyle{siam}

\end{document}